\documentclass[12pt]{article}

\usepackage{epsfig}

\usepackage{amssymb}
\usepackage{amsfonts}

\usepackage{color}
 
\def\be{\begin{equation}}
\def\ee{\end{equation}}
\def\ba{\begin{array}{c}}
\def\ea{\end{array}}

\def\ben{$$}
\def\een{$$}

\newcommand{\bea}{\begin{eqnarray}}
\newcommand{\eea}{\end{eqnarray}}

\newcommand{\kt}{\rangle}

\newtheorem{thm}{Theorem}

\newtheorem{lemma}[thm]{Lemma}

\newtheorem{defn}[thm]{Definition}

\newenvironment{proof}{\noindent
 {\bf Proof.}}{\hfill$\square$\vspace{3mm}\endtrivlist}

\begin{document}

\begin{center}

{\Large \bf

Quantum phase transitions mediated by
clustered non-Hermitian degeneracies

}

\vspace{0.8cm}

  {\bf Miloslav Znojil}

\vspace{0.2cm}

The Czech Academy of Sciences, Nuclear Physics Institute,

 Hlavn\'{\i} 130,
250 68 \v{R}e\v{z}, Czech Republic

\vspace{0.2cm}

 and

\vspace{0.2cm}

Department of Physics, Faculty of Science, University of Hradec
Kr\'{a}lov\'{e},

Rokitansk\'{e}ho 62, 50003 Hradec Kr\'{a}lov\'{e},
 Czech Republic

\vspace{0.2cm}

{e-mail: znojil@ujf.cas.cz}

\end{center}



\section*{Abstract}

The phenomenon of degeneracy
of an $N-$plet of bound states
is studied in the framework
of the quasi-Hermitian
(a.k.a. ${\cal PT}-$symmetric)
formulation of quantum
theory of closed systems.
For a general non-Hermitian Hamiltonian $H=H(\lambda)$
such a degeneracy may occur at a real
Kato's exceptional point
$\lambda^{(EPN)}$ of order $N$
and of the geometric multiplicity
{\it alias\,} clusterization index
$K$.
The corresponding
unitary
process of
collapse (loss of observability)
can be then interpreted as
a generic
quantum phase transition.
The dedicated literature deals, predominantly,
with the non-numerical benchmark models
of the simplest processes
where $K=1$.
In our present paper it is shown that in the ``anomalous''
dynamical scenarios with
$1< K \leq N/2$ an analogous
approach is applicable.
A multiparametric
anharmonic-oscillator-type
exemplification of such systems
is constructed as a set of
real-matrix $N$ by $N$ Hamiltonians which
are exactly solvable,
maximally non-Hermitian and
labeled by
specific {\it ad hoc\,}
partitionings  ${\cal R}(N)$ of $N$.

\newpage

\section{Introduction}

The experimentally highly relevant phenomenon of a
quantum phase transition
during which at least one of the
observables loses its observability status
is theoretically elusive. In the conventional
non-relativistic Schr\"{o}dinger picture, for example,
the observability of the energy
represented by a self-adjoint
Hamiltonian $H=H(\lambda)$ appears too robust
for the purpose. In phenomenological applications,
for this reason,
the onset of phase transitions
is simulated, typically,
by an abrupt, discontinuous change of
the operators
at $\lambda=\lambda^{(critical)}$
\cite{Messiah,Landau}.

In relativistic quantum mechanics the situation is, paradoxically, under a
better theoretical control.
In the Klein-Gordon equation with Coulomb potential, for example,
the energy levels can merge and complexify at a finite,
{\em dynamically\,} determined critical
strength $\lambda^{(critical)}$ of the attraction \cite{Greiner}.
The operator representing energy remains unchanged.
One must only reconstruct a correct, sophisticated,
Hamiltonian-dependent
physical Hilbert space of states ${\cal H}={\cal H}(\lambda)$
offering a consistent
probabilistic interpretation of the system
up to $\lambda^{(critical)}$ \cite{aliKG}.

The latter result is one of applications of the
recent innovative formulation of quantum theory
in which a preselected (i.e., relativistic as well as
non-relativistic) Hamiltonian $H=H(\lambda)$ need not be
self-adjoint (concerning this theory we shall add more details
below; preliminarily, interested readers may consult, say, one of
reviews \cite{Geyer,Carl,SIGMA,ali,book}). This means that even for
non-Hermitian Hamiltonians (with real spectra)
and even in the dynamical regime close to
$\lambda^{(critical)}$, there exists
a feasible strategy of making
the evolution of the corresponding quantum system unitary.

In our recent paper \cite{corridors} we described a few
basic aspects of implementation
of the latter model-building strategy
in the case of quantum systems close to their
phase transition.
The key technical
aspects
of the theory
were illustrated using the
$N$-by-$N-$matrix Hamiltonians ${H}^{(N)}(\lambda)$
for which the critical
parameter
can be identified with a Kato's
\cite{Kato} exceptional point
of order $N$ (EPN,  $\lambda^{(critical)}\,\equiv\,\lambda^{(EPN)}$).
In such an arrangement
Hamiltonians  ${H}^{(N)}(\lambda)$
possessed the real and non-degenerate spectra at
$\lambda< \lambda^{(EPN)}$ while
exhibiting,
in the EPN limit,
the characteristic phase-transition behavior (a detailed
explanation will be given
in section \ref{statarta} below).

The method used in paper \cite{corridors}
was perturbative so that the
specification of
the eligible non-Hermitian Hamiltonians $H(\lambda)$ was merely
indirect. Thus, the results were complemented by
a subsequent technical
paper \cite{passage} in which several
closed-form Hamiltonians were
described. Unfortunately,
in
another, purely numerical study of the
phase-transition problem \cite{anomalous}
it has been revealed
that from the mathematical as well as
experiment-oriented points of view
the class of models considered in \cite{corridors}
and \cite{passage}
must be declared too narrow.
In the light of this observation the tone of the
overall conclusions
was rather discouraging.
The existence of
``not quite expected technical subtleties''
has been emphasized, with the
``word of warning \ldots supported by an explicit ill-behaved
illustrative matrix model'' \cite{anomalous}.

In our present paper
we will
prolong the latter studies
but, first of all,
we will strongly
oppose
the scepticism of their conclusions.
We will, first of all,
broaden the class of
Hamiltonians in a way which will
fill the gaps (see section \ref{pemm}).
As an unexpected mathematical byproduct of these efforts,
an unusual exhaustive
combinatorial classification of our class of
models
will be formulated and
summarized in Appendix A.

In section \ref{specide}
we shall show that the family of
our present solvable
anharmonic-oscillator-type
benchmark
models
covers
all of the mathematically admissible
realizations of the EPN-related
quantum phase transitions.
An explicit sample of our exhaustive classification pattern
will be presented, in section \ref{seduma},
up to $N=8$.
The overall discussion
(emphasizing, e.g., that
our classification is non-numerical,
circumventing all of the
above-mentioned ill-conditioning dangers)
and a concise summary
will be finally added in
sections \ref{speciae} and \ref{summary}.

\section{Quantum phase transitions\label{statarta}}

The main weakness of the models
of quantum phase transitions
as presented
in our preceding papers \cite{corridors} and \cite{passage}
lies in a restriction of their scope to
a fairly small subset of the mathematically
admissible unitary evolution scenarios.
An explanation of this restriction
becomes facilitated when one turns attention to the
models which are exactly solvable.

\subsection{Exceptional points of order $N$ in solvable non-Hermitian models}

In the simplest example taken from \cite{passage}
the overall discussion
of the phase transition processes was based on a detailed analysis of
the  $N$ by $N$
tridiagonal-anharmonic-oscillator (TAO)
Hamiltonians
 \be
 H^{(N)}_{\rm (TAO)}(\lambda)
 =\left[ \begin {array}{ccccc}
1-N&b_1(\lambda)
  &0&\ldots&0
  \\{}-b_1(\lambda)&3-N
  &\ddots
 &\ddots&\vdots
 \\{}0&-b_2(\lambda)
 &\ddots&b_2(\lambda)&0
 \\{}\vdots&\ddots&\ddots&N-3&b_{1}(\lambda)
 \\{}0&\ldots&0&-b_{1}(\lambda)&N-1
 \end {array} \right]\,.
 \label{pentoy}
 \ee
Such a model can be interpreted,
after an inessential shift of the
origin of the energy scale, as an
antisymmetric-matrix perturbation
of a
truncated diagonal matrix form of harmonic oscillator.
In the weak coupling limit the spectrum is
known of course. In the opposite,
strongly anharmonic dynamical regime the
localization of the
spectrum becomes, in general, numerical.
The Hamiltonian becomes dominated
by its off-diagonal part which is
chosen real, antisymmetric (i.e., maximally non-Hermitian)
and,
for the reasons explained in \cite{maximal},
${\cal PT}-$symmetric,
i.e., symmetric with respect to the second diagonal.

For the purposes of study
of the mechanisms of a unitary passage of quantum systems
through their EPN singularities the latter tridiagonal models
proved particularly suitable because
in a broad range of matrix element functions $b_j(\lambda)$
their spectra remained real.
After a more restricted choice
of functions $b_j(\lambda)$
these spectra appeared real up to
$\lambda = \lambda^{(EPN)}$
(cf. \cite{tridiagonal}).
At the critical value of
$\lambda^{(EPN)}$ all of the $N$ energy levels merged
while they ceased to be real beyond $\lambda^{(EPN)}$.

From our present point of view
the most relevant feature of the model
is that at any integer $N\geq 2$
the Hamiltonians degenerate,
in the phase-transition EPN limit, to the respective
non-diagonalizable but elementary,
closed-form $N$ by $N$ matrices
 \be
 H^{(2)}_{\rm (TAO)}(\lambda^{(EP2)})
 = \left [\begin {array}{cc} -1&1\\{}-1&1
 \end {array}\right
 ]\,,
 \ \ \ \ \ \
 H^{(3)}_{\rm (TAO)}(\lambda^{(EP3)})
  = \left [\begin {array}{ccc} -2&\sqrt{2}&0\\{}-\sqrt{2}&0
 &\sqrt{2}\\{}0&-\sqrt{2}&2\end {array}\right ],
 \label{epthisset}
 \ee
 \ben
 H^{(4)}_{\rm (TAO)}(\lambda^{(EP4)}) = \left [\begin {array}{cccc}
  -3&\sqrt{3}   &0  &0\\
 -\sqrt{3}&-1   &2  &0\\
  0&-2  &1 &\sqrt{3}\\
  0&0&-\sqrt{3}&3
 \end {array}\right ]\,,
 \ \ \ 
 H^{(5)}_{\rm (TAO)}(\lambda^{(EP5)})
 =\left [\begin {array}{ccccc} -4&2&0&0&0\\{}-2&
 -2&\sqrt{6}&0&0\\{}0&-\sqrt{6}&0&\sqrt{6}&0
 \\{}0&0&-\sqrt{6}&2&2\\{}0&0&0&-2& 4\end {array}\right ]\,
 \een
etc \cite{passage}. The respective spectra degenerate to 
a single real 
value,
 \be
 \lim_{\lambda \to \lambda^{(EPN)}}\,E_n(\lambda)=\eta\,,\ \ \ \
 n=1,2,\ldots,N
  \,.
 \label{siesta}
 \ee
This is precisely the mathematical feature of the spectrum
which finds its physical interpretation of an instant of quantum
phase transition \cite{BM}.

In our particular models the value of $\eta=0$
remains $N-$independent. The complete EPN
degeneracy of energies (\ref{siesta})
proves accompanied by the complete degeneracy of
all of the related eigenvectors,
 \be
 \lim_{\lambda \to \lambda^{(EPN)}}\,|\psi_{n}^{(N)}(\lambda)\kt
 =|\chi^{(N)}(\lambda)\kt\,,\ \ \  \ \
 \,,\ \ \ \
 n=1,2,\ldots,N
  \,.
 \label{kuwinde}
 \ee
For our forthcoming analysis of the
mechanisms of the limiting loss-of-the-observability transition
it will be vital to know that
all of the matrix limits (\ref{epthisset})
can easily be transformed to their  unique,
``canonical'', $N$ by $N$
Jordan-matrix respective forms
 \be
 J^{\rm (N)}(\eta)
 =\left[ \begin {array}{ccccc}
  \eta&1&0&\ldots&0
  \\{}0&\eta&1 &\ddots&\vdots
 \\{}0&\ddots&\ddots&\ddots&0
 \\{}\vdots&\ddots&0&\eta&1
 \\{}0&\ldots&0&0&\eta
 \end {array} \right]\,.
 \label{heuu}
 \ee
The transformations are
mediated by the Schr\"{o}dinger-like
equation
 \be
 H^{(N)}_{(TAO)}(\lambda^{(EPN)})\,Q^{(N)}
 =Q^{(N)}\,J^{\rm (N)}(\eta)\,.
 \label{kano}
 \ee
For models (\ref{epthisset}), all of
the transition matrices $Q^{(N)}$ 
defined
by this equation are available 
in closed form
(see \cite{passage}).

\subsection{Numerical models and clustered non-Hermitian degeneracies}


The empirical observations
published in the
strictly numerically oriented paper  \cite{anomalous}
indicate that
the tridiagonal-matrix
choice of models (\ref{pentoy}) may be over-restrictive.
Such a suspicion results from the completeness
of the degeneracy (\ref{kuwinde}) of the eigenvectors. Indeed,
from the point of view of linear algebra
or functional analysis
such a type of degeneracy represents
a special case \cite{Kato}.
In general one can only expect a weaker form
of such a degeneracy in which the
eigenvectors $|\psi_{n}^{(N)}(\lambda)\kt$ of the Hamiltonian
would form a $K-$plet
of clusters of the degenerating eigenvectors with the set of
subscripts $n=1,2,\ldots,N$ decomposed into $K$ non-overlapping
subsets $S_k$,
 \be
 \lim_{\lambda \to \lambda^{(EPN)}}\,|\psi_{n_k}^{(N)}(\lambda)\kt
 =|\chi_k^{(N)}(\lambda)\kt\,,\ \ \  \ \
 n_k \in S_k\,,\ \ \
 \ \
 k=1,2,\ldots,K\,.
 \label{kwinde}
 \ee
In the light of this
observation the results of papers \cite{corridors,passage}
must be reinterpreted as covering only a
subfamily of all of the possible
EPN-related quantum
phase transitions.
Any complete set of
benchmark Hamiltonians
$H^{(N)}(\lambda)$ must contain
matrices which are more general than tridiagonal.
In what follows our attention will be, therefore, redirected
to the full real matrices
of an analogous, antisymmetrically perturbed
general anharmonic oscillator (GAO)
form
 \be
 H^{(N)}_{\rm (GAO)}(\lambda)
 =\left[ \begin {array}{cccccc}
  1-N&b_1(\lambda)&c_1(\lambda)
  &d_1(\lambda)&\ldots&\omega_1(\lambda)
  \\{}-b_1(\lambda)&3-N
  &b_2(\lambda)
 &c_2(\lambda)&\ddots&\vdots
 \\{}-c_1(\lambda)&\ddots
 &\ddots&\ddots&\ddots&d_1(\lambda)
 \\{}-d_1(\lambda)&\ddots&-b_3(\lambda)
 &N-5&b_{2}(\lambda)&c_1(\lambda)
 \\{}\vdots&\ddots&-c_2(\lambda)&-b_{2}(\lambda)&N-3&b_1(\lambda)
 \\{}-\omega_1(\lambda)&\ldots&-d_1(\lambda)
 &-c_1(\lambda)&-b_{1}(\lambda)&N-1
 \end {array} \right]\,.
 \label{pefull}
 \ee
Our present paper could be read, in this light, as a
continuation and as an ultimate completion
of the project of Ref.~\cite{anomalous}.
{\it A priori}, one might be
critical towards such a project.
Indeed, from the numerical experiments
as performed in \cite{anomalous} it can be deduced
that
at the larger dimensions $N$
the work with tridiagonal Hamiltonians
is the only constructively
tractable
option.
In other words,
the study of the
EPN-supporting toy-model
Hamiltonians
$H^{(N)}_{\rm}(\lambda)$
should work either with the general tridiagonal
matrices as sampled above,
or with some of their non-tridiagonal
generalizations defined at
a few smallest matrix dimensions $N$
(thus, for example, the
``easily tractable'' maximum was
found at $N=6$ in \cite{anomalous}).
In our present paper
we will just describe a new, third option showing that
an amended and universal EPN-related model-building
strategy does
exist. 

Our results will be based on the use of 
non-tridiagonal
matrices which are sparse, rendering
the implementation of the strategy
feasible in applications.
The presentation
of the idea may start from the
replacement of the too special condition of
Eq.~(\ref{kano}) by the generalized, standard relation
 \be
 H^{(N)}_{(GAO)}(\lambda^{(EPN)})\,Q^{(N)}
 =Q^{(N)}\,{\cal J}^{[{\cal R}(N)]}(\eta)\,.
 \label{kanon}
 \ee
The superscript ${\cal R}(N)$ denotes here one of the
partitions of $N=N_1+N_2+\ldots +N_K$
that do not contain 1 as a part
(and are such that, say, $N_1\geq N_2\geq \ldots \geq N_K\geq 2$;
see \cite{A002865} or Table~\#1
in \cite{degdeg}).
The integer $K$
represents the above-mentioned clusterization index {\it alias\,}
geometric multiplicity of the EPN degeneracy \cite{Kato}.

As long as the first two
partitions ${\cal R}(2)=2$ and ${\cal R}(3)=3$
are unique (for both of them the geometric
EPN multiplicity $K$ is equal to one),
equations~(\ref{kano}) and (\ref{kanon})
remain the same at $N=2$ and $N=3$. The difference
reflecting the existence of non-trivial
multiplicities $K>1$
only emerges at $N=4$
where we can have ${\cal R}_1(1)=4$
(i.e., $K=1$) and ${\cal R}_2(4)=2+2$
(i.e., $K=2$).
Thus, in our present
amended model-building recipe we choose any $N\geq 4$,
pick up one of the
partitionings ${\cal R}(N)$,
skip the $K=1$ cases (which are well known)
and
define the block-diagonal matrix ${\cal J}^{[{\cal R}(N)]}(\eta)$
in the form of
direct sum of a $K-$plet of
elementary Jordan blocks,
 \be
  {\cal J}^{[{\cal R}(N)]}(\eta)= J^{(N_1)}(\eta) \oplus
  J^{(N_2)}(\eta) \oplus
  \ldots \oplus
  J^{(N_K)}(\eta)\,.
  \label{jobl}
 \ee
The first
alternative options
emerge at $N=4 = 2+2$ and at $N=5=3+2$, with the following two new,
$K=2$ direct sums
of the Jordan blocks,
 \be
 {\cal J}^{[2+2]}(\eta)
 =
 \left[ \begin {array}{cc|cc}
  \eta&1&0&0
  \\{}0&\eta&0
 &0
 \\\hline
 {}0&0&\eta&1
 \\{}0&0&0&\eta
 \end {array} \right]
 \,
 \ \ \ \
 {\rm and}\
 \ \ \
 {\cal J}^{[3+2]}(\eta)
 =
 \left[ \begin {array}{ccc|cc}
  \eta&1&0&0&0
  \\{}0&\eta&1&0
 &0
 \\
 {}0&0&\eta&0&0
 \\ \hline
 {}0&0&0&\eta&1
 \\{}0&0&0&0&\eta
 \end {array} \right]
 \,.
 \label{feuu}
 \ee
We are now prepared to turn attention to the
construction of multi-diagonal GAO models
with the full-matrix structure (\ref{pefull})
of their Hamiltonians,
and with the general
direct-sum structure (\ref{jobl})
of their canonical representation in the EPN limit.

\section{Systems with pentadiagonal-matrix Hamiltonians\label{pemm}}

\subsection{Elementary one-parametric model}

A few simulation of dynamics near EPNs
using multidiagonal $N$ by $N$ matrix Hamiltonians
were presented in paper \cite{anomalous}.
The scope of the study was restricted,
due to the apparently purely numerical nature
of the problem, to
the smallest matrix dimensions $N \leq 6$.
Such a restriction helped to keep
the necessary evaluations of the spectra
non-numerical.
Incidentally, the latter decision was a bit unfortunate
because, as we will see below, the next option with $N=7$
would have been perceivably more instructive.
Still, the key message of the study
remains significant:
the search for anomalous $K>1$
EPN singularities
with optional
geometric multiplicities
should be based on a systematic analysis of
non-tridiagonal, multi-diagonal matrix models.

In the light of this experience let us now turn attention
to the following pentadiagonal-matrix example
with $N=7$,
 \be
 {H^{\rm (toy)}}(\lambda)=  \left[ \begin {array}{ccccccc}
  1&0&\sqrt {3}g&0&0&0&0
 \\\noalign{\medskip}0&3&0&\sqrt {2}g&0&0&0
 \\\noalign{\medskip}-\sqrt {
3}g&0&5&0&{2}g&0&0
\\\noalign{\medskip}0&-\sqrt {2}g&0&7&0&\sqrt {2}g&0
 \\\noalign{\medskip}0&0&-{2}g&0&9&0&\sqrt {3}g
 \\\noalign{\medskip}0&0&0&-
 \sqrt {2}g&0&11&0\\\noalign{\medskip}0&0&0&0&-\sqrt {3}g&0&13
 \end {array} \right]\,.
 \label{tomograf}
 \ee
The unperturbed truncated harmonic-oscillator
spectrum is kept unshifted, the
anharmonicity is
antisymmetric, and the freedom of the
$\lambda-$dependence
of the
off-diagonal elements of the perturbation
is reduced to a single
function $g = g(\lambda)$.
Such a simplification implies that
the related Schr\"{o}dinger bound-state problem
 \be
 {H^{\rm (toy)}}(g)\,|
 {\psi_n}(g)\kt=E_n(g)\,| {\psi_n}(g)\kt\,
 \label{themo}
 \ee
becomes solvable exactly,
 $$
 E_0(g)=7\,,\ \ \ \ \
 E_{\pm 1}(g)
 =7\pm \sqrt {4-{g}^{2}}
$$
 $$
 E_{\pm 2}(g)
 =7\pm 2\,\sqrt {4-{g}^{2}}
\,,\ \ \ \
 E_{\pm 3}(g)
 =7\pm 3\,\sqrt {4-{g}^{2}}\,.
$$
The model exemplifies the system
in which the conventional
self-adjoint harmonic-oscillator dynamics is
realized
at $g=0$, and in which the small perturbations,
in spite of their maximal non-Hermiticity,
keep the spectrum real and, hence, observable,
in principle at least.

\subsubsection{Strong-coupling dynamical regime}

The above-listed explicit formulae demonstrate that
the whole spectrum of our toy model remains real up to
the EP7 limit of $g \to g^{(EP7)}=2$.
The bound-state energies remain real and
well separated
along a path connecting
the weakly-anharmonic (WA)  and
the strong-coupling (SC)
ends of
the
open interval of the values of
$g \in (0,2)$.
Even when one decides to consider
a nontrivial function $g=g(\lambda)$
of parameter $\lambda$ which would not deviate
too much from the linear one,
one can still deduce that there exists a
fairly broad corridor of unitarity
which connects the
harmonic-oscillator (HO) and the EP7 dynamical extremes.

In the
WA regime 
the optional auxiliary (and, say, monotonously increasing)
function $g(\lambda)$ is to be kept
small. Then, the
anharmonicity will remain
easily tractable by the standard
Rayleigh-Schr\"{o}dinger perturbation methods.
Near the opposite SC boundary where $g \lessapprox 2$
the EPN degeneracy (\ref{siesta})
is reached.
The reality of the spectrum becomes mathematically fragile.
Whenever the value of the coupling exceeds its critical
value of $g^{(EPN)}=2$, the whole spectrum becomes complex.

The latter form of the EPN-related
instability
is a real, measurable phenomenon. 
Its possible detection appeared to be a true challenge
during the popular
experimental simulations of quantum dynamics
via non-quantum systems
(cf., e.g., extensive reviews
of this point in \cite{Carlbook,Christodoulides}).
Fortunately, it has recently been clarified
that in the genuine closed quantum systems
living in an EPN vicinity
an experimental realization
of the similar instabilities in
the laboratory would be
much more complicated if not 
even impossible \cite{admissible}.
The subtle reasons of the existence of such a paradox
were briefly explained in \cite{Ruzicka}.
Their essence lies in the fact that
the
perturbations which would make
the system leave the physical Hilbert space ${\cal H}$
would be hard to define
in practically any quantum theory of closed systems
including not only its conventional textbook forms
but, equally well,
all of its various quasi-Hermitian \cite{Geyer},
pseudo-Hermitian \cite{ali}
or ${\cal PT}-$symmetric \cite{Carl} versions.
In the latter setting, indeed,
the
Hamiltonian-dependent choice of the
physical Hilbert space ${\cal H}$ is ambiguous \cite{ali}.
For this reason, {\em any\,} change of the
Hamiltonian reopens
the ambiguity problem
and, in this sense, makes the perturbation theory
nonlinear \cite{degdeg,Ruzicka}.

Our present model may be recalled for illustration purposes.
In it, we are allowed to
introduce a new small
parameter $\kappa=\kappa(\lambda)\in (0,1)$ and to
redefine $g=\widetilde{g}(\kappa)=2\,(1-\kappa^2)$.
This enables us to consider the
related (``tilded'') modification of our
spectral problem (\ref{themo}) with
the same exact eigenvalues
rewritten
in an equivalent but SC-friendlier form
 $$
 \widetilde{E}_0(\kappa)=7\,,\ \ \ \ \
 \widetilde{E}_{\pm 1}(\kappa)
 =7 \pm 2\,\sqrt {-{\kappa}^{4}+2\,{\kappa}^{2}}
 \sim 7 \pm 2\,\sqrt {2}\,{\kappa}+ {\cal O}(\kappa^3)\,,
 $$
  $$
  \widetilde{E}_{\pm 2}(\kappa)=7 \pm 4\,\sqrt {-{\kappa}^{4}
  +2\,{\kappa}^{2}}
 \,,\ \ \ \ \
  \widetilde{E}_{\pm 3}(\kappa)
  =7 \pm 6\,\sqrt {-{\kappa}^{4}+2\,{\kappa}^{2}}\,.
 $$
Also in this representation these eigenvalues remain all real
at small $\kappa$,
reconfirming the existence of a corridor of unitarity
connecting the WA and SC dynamical-regime extremes.

%
%
%
%

\subsubsection{Canonical form of the Hamiltonian}

In the SC EP7 limit $\kappa \to 0$
the spectrum becomes degenerate and the
Hamiltonian itself ceases to be diagonalizable.
It may be shown to possess just two eigenvectors
so that
the EPN value $\eta=7$ of the energy can be used in
the canonical eigenvalue problem (\ref{kanon})
where
 \be
 {\cal J}^{\rm (4+3)}(\eta)
 =\left[ \begin {array}{cccc|ccc}
  \eta&1&0&0&0&0&0
  \\{}0&\eta&1
 &0&0&0&0
 \\{}0&0&\eta&1&0&0&0
 \\{}0&0&0&\eta&0
 &0&0
 \\
 \hline
 0&0&0&0&\eta&1&0
 \\{}0&0&0&0&0&\eta&1
 \\{}0&0&0&0&0&0&\eta
 \end {array} \right]
 =
 \left[ \begin {array}{cccc}
  \eta&1&0&0
  \\{}0&\eta&1
 &0
 \\{}0&0&\eta&1
 \\{}0&0&0&\eta
 \end {array} \right]
 \oplus
 \left[ \begin {array}{ccc}
  \eta&1&0
  \\{}0&\eta&1
 \\{}0&0&\eta
 \end {array} \right]
 \,.
 \label{euu}
 \ee
This is a direct sum
of
two Jordan-block matrices.
The partitioning
merely guides the eye and emphasizes the fact
that in our toy-model Hamiltonian
(\ref{tomograf})
there is no mutual coupling between
the even and odd indices. In the EP7 limit,
our Hamiltonian $H^{\rm (toy)}(\lambda^{(EP7)})$
may be interpreted
as a direct sum of the two independent components,
 \be
  H^{\rm (toy)}(\lambda^{(EP7)})
  =H^{[odd]}_{(EP7)}\oplus H^{[even]}_{(EP7)}
  \label{limits}
  \ee
where
 $$
  H^{[odd]}_{(EP7)}=
 \left[ \begin {array}{cccc}
 1&2\,\sqrt {3}&0&0
  \\-2\, \sqrt {3}&5&4&0
 \\0&-4&9&2\,\sqrt {3}
 \\0&0
 &-2\,\sqrt {3}&13
 \end {array} \right]\,,
 \ \ \ \ \
  H^{[even]}_{(EP7)}=
 \left[ \begin {array}{ccc}
 3&2\,\sqrt {2}&0
  \\-2\, \sqrt {2}&7&2\,\sqrt {2}
 \\0
 &-2\,\sqrt {2}&11
 \end {array} \right]
 \,.
 $$
Off the EP7 limit, our specific pentadiagonal
matrix $H^{\rm (toy)}(\lambda)$ still has
the form of a direct sum of the two tridiagonal
matrices.
Due to our specific choice of the model, they are both
exactly solvable so that the whole $N=7$ model is
solvable in closed form as well.
Moreover,
arbitrary {\it ad hoc\,} perturbation terms may be added to
couple the components of the direct sum -- see, e.g.,
an application of such an idea to
a realistic model of such a type in \cite{without}.

A
specific consequence of the simplicity
of our present perturbed model
lies in
the
availability of the explicit  transition-matrix
solution
of Eq.~(\ref{themo}),
 $$
 Q^{\rm (toy)}= \left[ \begin {array}{ccccccc} -48&24&-6&1&0&0&0\\{}0
&8&-4&1&8&-4&1\\{}-48\,\sqrt {3}&16\,\sqrt {3}&-2\,
\sqrt {3}&0&0&0&0\\{}0&8\,\sqrt {2}&-2\,\sqrt {2}&0&8
\,\sqrt {2}&-2\,\sqrt {2}&0\\{}-48\,\sqrt {3}&8\,
\sqrt {3}&0&0&0&0&0\\{}0&8&0&0&8&0&0
\\{}-48&0&0&0&0&0&0\end {array} \right]\,.
 $$
Indeed, in the
perturbation constructions of the SC states
the role of
unperturbed basis
is relegated, in natural manner, to transition matrices
(see more details in \cite{degdeg,without}).

\subsection{Multiparametric Hamiltonians}

The tricks used in connection with the
one-parametric toy model (\ref{tomograf}) can immediately be
applied to the  $N-$dimensional
model
 \be
 H^{(N)}_{\rm (pent. special)}(\lambda)
 =\left[ \begin {array}{c|c|c|c|c|c}
  1-N&0&c_1(\lambda)
  &0&\ldots&0
  \\ \hline0&3-N
  &0
 &\ddots&\ddots&\vdots
 \\ \hline-c_1(\lambda)&0
 &\ddots&\ddots&c_2(\lambda)&0
 \\ \hline0&\ddots&\ddots&N-5&0&c_1(\lambda)
 \\ \hline\vdots&\ddots&-c_2(\lambda)&0&N-3&0
 \\ \hline0&\ldots&0&-c_1(\lambda)&0&N-1
 \end {array} \right]\,
 \label{cepez}
 \ee
with an elementary shift of the origin on the energy scale
and with
the parallels in structure
emphasized by the partitioning.
Obviously, all of these matrices are
equal to direct sums of
two tridiagonal matrices, viz.,
 \be
 H^{(N)}_{\rm (component\ one)}(\lambda)
 =\left[ \begin {array}{cccc}
  1-N&c_1(\lambda)
  &0&\ldots
 \\ -c_1(\lambda)&5-N
 &c_3(\lambda)&\ddots
 \\ 0&-c_3(\lambda)&9-N
 &\ddots
 \\ \vdots&\ddots&\ddots&\ddots
 \end {array} \right]\,
 \label{apezerb}
 \ee
and
 \be
 H^{(N)}_{\rm (component\ two)}(\lambda)
 =\left[ \begin {array}{cccc}
  3-N&c_2(\lambda)
  &0&\ldots
 \\ -c_2(\lambda)&7-N
 &c_4(\lambda)&\ddots
 \\ 0&-c_4(\lambda)&11-N
 &\ddots
 \\ \vdots&\ddots&\ddots&\ddots
 \end {array} \right]\,.
 \label{bepeze}
 \ee
Due to our assumption of maximal
non-Hermiticity (i.e., of antisymmetry) of
perturbations
(i.e., of anharmonicities),
their diagonal component will always
vanish.
Thus, without confusion, we may refer to the matrix of
Eq.~(\ref{cepez}) by
its main diagonal put in the box,
$\fbox{1-N,3-N,\ldots,N-1}$.
Its direct-sum decomposition into
components (\ref{apezerb}) and (\ref{bepeze})
can be then
written in shorthand,
 $$
 \fbox{1-N,3-N,\ldots,N-1}=\fbox{1-N,5-N,9-N, \ldots} \oplus
 \fbox{3-N,7-N,11-N,\ldots}\,.
 $$
The last elements of the summands are not displayed because they
vary with
the parity of $N$.
After the explicit specification of
the parity of $N$
we
arrive at the following two conclusions.

\begin{lemma}
At the even matrix dimension $N=2J$,
the decomposition
of the pentadiagonal
sparse-matrix model (\ref{cepez})
into its
tridiagonal TAO components (\ref{apezerb}) and (\ref{bepeze})
only supports the two $K=1$ EPJ limits (\ref{epthisset}),
with different respective
energies $\eta = \pm 1$. At any one of them, the confluence
in Eq.~(\ref{siesta}) is incomplete, involving just $J$ levels.
\label{lemma1}
\end{lemma}
\begin{proof}
The main diagonal
$\fbox{1-N,3-N,\ldots,N-1}= \fbox{1-2J,3-2J,\ldots,2J-1}$
of matrix (\ref{cepez}) does not contain a central zero.
The central interval $(-1,1)$
is ``too short''. Its two elements $-1$ and $1$
get distributed among both of the
components (\ref{apezerb}) and (\ref{bepeze}). In the
resulting direct sum
  $$
  \fbox{1-N,3-N,\ldots,N-1}= 
  \fbox{1-2J,5-2J, \ldots,2J-3}\oplus
  \fbox{3-2J,7-2J,\ldots,2J-1}\,,
  $$
both of the components will be centrally asymmetric.
\end{proof}

 \noindent
The search for an anomalous EPN with $K=2$ failed.
Even after a successful $J$ by $J$ realization
of the two separate EPJ limits using building blocks (\ref{epthisset}),
requirement~(\ref{siesta}) will only offer
two different
values of the
eligible limiting EPN energies.
The direct-sum decomposition yields the two
non-anomalous $K=1$ EPNs
of the same small order  $J=N/2$. A better result is the next one.

\begin{lemma}
At odd $N=2J+1$, both of the
tridiagonal matrices (\ref{apezerb}) and (\ref{bepeze})
admit the respective realizations (\ref{epthisset})
of their EPN limits.
The related energies coincide so that
the direct sum (\ref{cepez})
admits the anomalous
EPN limit with
geometric multiplicity two.
\label{lemma2}
\end{lemma}

\begin{proof}
We have
  \be
  \fbox{1-N,3-N,\ldots,N-1}= \fbox{-2J,2-2J,\ldots,2J}=
  \fbox{-2J,4-2J, \ldots,2J}\oplus
  \fbox{2-2J,6-2J,\ldots,2J-2}
  \label{ledva}
  \ee
so that out of
the central triplet of integers $(-2,0,2)$,
the doublet  $(-2,2)$ remains long enough to be
a component of one of the sub-boxes.
Their respective dimensions $J+1$ and $J$
are now different.
This is compensated by the central symmetry
of the summands and by
the coincidence of the
EPN energies,
$\eta_\pm =0$.
In
the direct sum (\ref{cepez})
the respective EPN limits
degenerate to a single,  anomalous
EPN limit. The $K=2$ clusterization
(\ref{kwinde}) takes place.
\end{proof}

The highly plausible one-to-one correspondence
between the
tridiagonality of the Hamiltonian and the $K=1$
form of its EPN limit as
conjectured in \cite{anomalous}
is now complemented
by our two Lemmas which confirm
that in a search for models with larger $K$
even the pentadiagonality assumption
need not help too much.
Nevertheless, the general pentadiagonal models
 \be
 H^{(N)}_{\rm (pentadiagonal)}(\lambda)
 =\left[ \begin {array}{cccccc}
  1-N&b_1(\lambda)&c_1(\lambda)
  &0&\ldots&0
  \\{}-b_1(\lambda)&3-N
  &b_2(\lambda)
 &\ddots&\ddots&\vdots
 \\{}-c_1(\lambda)&-b_2(\lambda)
 &\ddots&\ddots&c_2(\lambda)&0
 \\{}0&\ddots&\ddots&N-5&b_{2}(\lambda)&c_1(\lambda)
 \\{}\vdots&\ddots&-c_2(\lambda)&-b_{2}(\lambda)&N-3&b_1(\lambda)
 \\{}0&\ldots&0&-c_1(\lambda)&-b_{1}(\lambda)&N-1
 \end {array} \right]\,
 \label{epen}
 \ee
offer another methodical inspiration
paving the way towards
the full-matrix scenario.
The idea is based on an additional
assumption that all of the matrix elements
$b_n(\lambda)$ are small. Then,
the decomposition
 \be
 H^{(N)}_{\rm (pentadiagonal)}(\lambda)
 = H^{(N)}_{\rm (pent. special)}(\lambda) + {\rm small \
 perturbations}
 \label{perpnt}
 \ee
could prove tractable by perturbation
techniques \cite{degdeg}.

\section{Multidiagonal solvable models\label{specide}}

The freedom of choice of any positive integer $K$
is desirable, but
the task is
ill-conditioned \cite{Wilkinson}.
The numerical localization of the EPN
degeneracies is difficult
at larger $N$s. This is well sampled,
e.g., in Ref.~\cite{bhgt6}.
We will only use non-numerical strategies in
our model-building project, therefore.

\subsection{Clusterization\label{duma}}

The core and essence of our forthcoming 
general non-numerical
constructions will lie in the mere generalization
of Eq.~(\ref{perpnt}), i.e., in 
an application of the idea that 
it makes sense to have 
some of the ``unfriendly'' GAO matrix elements 
re-classified as ``small perturbations''
(which could be, in the first run,  
neglected and omitted).
Such a reduction should help us to 
obtain a 
non-numerically tractable structure 
(analogous to matrix (\ref{cepez})) 
which could be factorized into a 
$K-$plet of solvable TAO components
(sampled, at $K=2$, by (\ref{apezerb}) and (\ref{bepeze})).

In a way
inspired by the 
pentadiagonal-matrix Lemma \ref{lemma2}
and, in particular, 
by the direct-sum decomposition (\ref{ledva}),
also the general GAO
full-matrix
Hamiltonian (\ref{pefull}) may still be 
identified and represented
by its
left-right antisymmetric main diagonal
put in a box,
 \be
 \fbox{1-N,3-N,\ldots,N-3,N-1}\,.
 \label{bosy}
 \ee
Formally, 
such a boxed symbol can be decomposed as follows,
 \be
 \fbox{1-N,3-N,\ldots,N-1}=
 \fbox{1-N,N-1}\,  \oplus  \fbox{3-N,5-N,\ldots,N-3}\,.
 \label{symbs}
 \ee
Naturally, the first right-hand-side component
$\fbox{1-N,N-1}$ of this decomposition could already 
represent the required
TAO type 
Hamiltonian matrix, provided only that
all of the ``unfriendly'' elements are omitted from
the outer rows and columns
of the initial Hamiltonian, yielding
 \be
 H^{(N)}_{\rm (spec.partit.)}(\lambda)
 =\left[ \begin {array}{c|cccc|c}
  1-N&0&0
  &\ldots&0&\omega_1(\lambda)
  \\ \hline 0&3-N
  &b_2(\lambda)
 &\ldots&z_2(\lambda)&0
 \\ 0&-b_2(\lambda)
 &\ddots&\ddots&\vdots&0
 \\ \vdots&\vdots&\ddots&N-5&b_{2}(\lambda)&\vdots
 \\ {} 0&-z_2(\lambda)&\ldots&-b_{2}(\lambda)&N-3&0
 \\ \hline -\omega_1(\lambda)&0&0&\ldots&0&N-1
 \end {array} \right]\,.
 \label{pepart}
 \ee
This enables us to decopose
 $$
 H^{(N)}_{\rm (spec.partit.)}(\lambda)
 =
 \left [(N-1)\times H^{(2)}_{\rm (toy)}(\lambda)\right ]
 \oplus H^{(N-2)}_{\rm (full)}(\lambda)\,
 $$
where the dimension of the second 
full-matrix component is diminished.
Thus, the construction of one of the possible direct-sum
decompositions
could be completed iteratively, with the ultimate
result 
preserving the TAO form of all of its components.

We are now prepared to
search for all of the other $K-$term
generalizations of the $K=2$
direct-sum expansion
(\ref{ledva}).
It is worth emphasizing that for the 
reasons illustrated in Lemma \ref{lemma1}
our fundamental methodical
requirement of the
non-numerical tractability of the EPN limits
of the GAO models
is in a one-to-one correspondence with the constraint
that all of
the 
components of
their $K-$term direct-sum expansions
must keep having the specific TAO form,
represented by the  
centrally antisymmetric boxed symbols.
Thus, an exhaustive classification of 
all of the possible direct-sum decompositions
of the initial symbol (\ref{bosy})
becomes an interesting
combinatorial problem with the solution
described in
Appendix A below.

\subsection{General direct-sum decompositions}

The unitary evolution scenarios
characterized by an incomplete,
``anomalous'' $K > 1$
degeneracy
of eigenstates
are all equally important \cite{anomalous}.
The present continuation of their analysis
will be inspired by Eq.~(\ref{limits}) where $K=2$.
We will fix
the parameter $\lambda=\lambda^{(EPN)}$
and assume 
an analogous GAO direct-sum decomposition
valid, in the EPN limit with specific $\eta=0$, at any
preselected dimension $N$ and multiplicity $K$,
 \be
  H^{\rm (N)}_{(GAO)}(\lambda^{(EPN)})
  =\widetilde{H^{(N_1)}}(\lambda^{(EPN)})\oplus
  \widetilde{H^{(N_2)}}(\lambda^{(EPN)})\oplus \ldots \oplus
  \widetilde{H^{(N_K)}}(\lambda^{(EPN)})\,.
    \label{limitsN}
  \ee
We should only keep in mind
that the value of the geometric multiplicity
$K$ cannot exceed $N/2$.
We will also insist on the non-numerical
tractability of the model.
Most easily, this goal will be achieved
by the requirement that
all of
the separate $N_j$-dimensional tilded
matrix components
of Hamiltonian (\ref{limitsN})
are elements of the above-mentioned TAO-Hamiltonian
family (\ref{epthisset}),
 \be
 \widetilde{H^{(N_j)}}(\lambda^{(EPN)})=
 c_j\,H^{(N_j)}_{(TAO)}(\lambda^{(EPN)})\,,\ \ \ \
 j = 1,2,\ldots,K\,.
 \label{rule0}
 \ee
The freedom of choice of $K$ different
normalization constants $c_j$
will not destroy the non-numerical form
and solvability of the model.
In a small vicinity of the EPN singularity
the
exact solvability of the TAO toy-models (\ref{pentoy})
will survive.
We may, therefore,
extend the definition of the model,
accordingly,
to parameters $\lambda < \lambda^{(EPN)}$
which do not lie too far from $\lambda^{(EPN)}$,
 \be
  H^{\rm (N)}_{(GAO)}(\lambda)
  =\widetilde{H^{(N_1)}}(\lambda)\oplus
  \widetilde{H^{(N_2)}}(\lambda)\oplus \ldots \oplus
  \widetilde{H^{(N_K)}}(\lambda)
  \ +  {\rm small\ corrections}\,.
    \label{limitsNK}
  \ee
For the sake of simplicity
let us ignore the corrections, and
let us only consider the
components with
weights
which remain $\lambda-$independent,
 \be
 \widetilde{H^{(N_j)}}(\lambda)=
 c_j\,H^{(N_j)}_{(TAO)}(\lambda)\,,\ \ \ \
 j = 1,2,\ldots,K\,.
 \label{rule0K}
 \ee
The parameter $\lambda$ may now decrease to zero
in a way which parallels the behavior of the tridiagonal
TAO models (\ref{pentoy}) of Refs.~\cite{passage,maximal,tridiagonal}.

During the process the
direct sum decomposition of the Hamiltonian
(i.e., the exact solvability of the
$K>1$ models~(\ref{limitsNK}) where we omitted
the ``correction'' term) will survive. Unfortunately,
for a general real $K-$plet of the
normalization constants $c_j$
the spectrum of model (\ref{limitsNK}) would be,
in the $\lambda \to 0$ limit,
non-equidistant.
This means that
in such a limit
our system would {\em not\,}
mimic the truncated harmonic oscillator.
This would be a truly unpleasant
feature of the model, especially in the context of
perturbation theory.
In our present paper the
equidistance of the unperturbed
$\lambda=0$ spectrum of the GAO model
will be, therefore, added as an independent WA postulate.

\begin{table}[h]
\caption{The list of all of the alternative TAO-direct-sum
decompositions (\ref{limitsNK})
of the GAO Hamiltonian (\ref{limits})
with symbol/label $\fbox{1-N,3-N,\ldots,N-1}$
at $N=6$ and $K>1$.}
 \label{dowe} \vspace{.4cm}
\centering
\begin{tabular}{||c|c|ccc|c||}
    \hline \hline
   \multicolumn{6}{||c||}{GAO label $\fbox{-5,-3,-1,1,3,5}$}\\
   \hline
   \hline
    $K$&${\cal R}(6)$&$j$&$N_j$&$c_j$&
    {\rm TAO}$_j$ label\\
\hline
\hline
  2&4+2&1&4&1& $\fbox{-3,-1,1,3}$\\
  &&2&2&5& $\fbox{-5,5}$\\
 \hline
  3&2+2+2&1&2&1& $\fbox{-1,1}$\\
  &&2&2&3& $\fbox{-3,3}$\\
  &&3&2&5& $\fbox{-5,5}$\\
  \hline
 \hline
\end{tabular}
\end{table}

The latter requirement will
restrict the freedom of our choice of the
normalization constants $c_j$
in Eq.~(\ref{rule0K}) of course. From the practical physical,
phase-transition-oriented point of view,
such a restriction appears 
acceptable, being rather severe only
at the
not too large integers $N$ and $K$.
This is illustrated in Table \ref{dowe}.
It shows that just two alternative
GAO models with $K>1$ will exist at $N=6$.
Nevertheless, the benefits of the
exact solvability of the restricted GAO models
will certainly prevail at the larger
matrix-dimensions because with the growth of
$N$ the number of the alternative scenarios
will grow very quickly: This growth
is sampled in Appendix A.

\section{Systematics of models with clustered EPN limits\label{seduma}}

Our model-building strategy is based on the
partitioning of an $N$ by $N$ GAO
Hamiltonian
labeled by the boxed main diagonal
(cf. Eq.~(\ref{bosy})) into a $K-$plet
of the TAO components represented by the shorter,
centrally antisymmetric boxed
equidistant subsets
 \be
 \fbox{$(1-N_j)\,c_j,(3-N_j)\,c_j,\ldots,(N_j-3)\,c_j,(N_j-1)\,c_j$}\,.
 \label{hjugobossy}
 \ee
This makes every candidate for a solvable benchmark Hamiltonian
equal to a direct sum of
TAO building blocks.
At the first few dimensions $N$,
the systematic
constructive implementation
of such a recipe will be made more explicit in
what follows.

\subsection{The choice of $N=2$ and $N=3$: no anomalous degeneracies}

For our present GAO class of
$\lambda-$dependent Hamiltonians (\ref{pefull})
there exists strictly one,
unique EP2 limit
satisfying our restrictions at $N=2$, namely, the
matrix $H^{(2)}_{\rm (TAO)}(\lambda^{(EPN)})$
as displayed in Eq.~(\ref{epthisset}).
In our abbreviated notation such a matrix is characterized by the
boxed symbol
$\fbox{-1,1}$. In the notation of Appendix A
the number $a(N)$ of eligible
scenarios is one, $a(2)=1$. In the EPN limit, the
geometric multiplicity of the spectrum is $K=1$.

Similarly, at $N=3$ we have
$a(3)=1$ and
the unique $K=1$
limit $H^{(3)}_{\rm (TAO)}(\lambda^{(EP3)})$
represented by the boxed
symbol $\fbox{-2,0,2}$ and
by the matrix
displayed in Eq.~(\ref{epthisset}).

\subsection{The simplest anomalous case with $N=4$ and $K=2$}

Besides the trivial $K=1$ option
with symbol $\fbox{-3,-1,1,3}$, there exists strictly one other
possibility of decomposition at $N=4$, viz.,
  $$
 \fbox{-3,-1,1,3}=\fbox{-1,1} \oplus \fbox{-3,3}
\,,\ \ \ \ K=2 \,.
 $$
In the limit $\lambda \to \lambda^{(EP4)}$
this direct sum represents the seven-diagonal but very
sparse GAO matrix
\be
 H^{(4)}_{(K=2)}(\lambda^{(EP4)}) = \left [\begin {array}{rrrr}
  -3&0   &0  &3\\
 0&1   &-1  &0\\
  0&-1  &1 &0\\
 -3&0&0&3
 \end {array}\right ]\,.
 \ee
In the unitarity domain where $\lambda \neq \lambda^{(EP4)}$
the number of the eligible 
dynamical scenarios is two, $a(4)=2$,
one of them representing a nontrivial clusterization
with $K=2$.

\subsection{Two $K=2$ options at $N=5$}

At $N=5$ the number of scenarios is three, $a(5)=3$.
Besides the trivial case we have
the two $K=2$ decompositions
  $
 \fbox{-4,-2,0,2,4}=\fbox{-2,0,2} \oplus \fbox{-4,4}
 $ and
 $
 \fbox{-4,-2,0,2,4}=\fbox{-4,0,4} \oplus \fbox{-2,2} $,
leading to the two 
alternative, non-equivalent wave function clusterizations.
These two alternatives are 
represented by the two
respective EP5 limiting matrix Hamiltonians, viz., by the
nine-diagonal
 $$
  H^{(5)}_{(K=2,a)}(\lambda^{(EP5)})
 =\left [\begin {array}{rrrrr}
  -4&0&0&0&4
 \\{}0& -2&\sqrt{2}&0&0
 \\{}0&-\sqrt{2}&0&\sqrt{2}&0
 \\{}0&0&-\sqrt{2}&2&0
 \\{}-4&0&0&0& 4
 \end {array}\right ]\,
 $$
and/or by the pentadiagonal
 $$
  H^{(5)}_{(K=2,b)}(\lambda^{(EP5)})
 =\left [\begin {array}{crccc}
    -4&0&2\sqrt{2}&0&0
 \\{}0& -2&0&2&0
 \\{}-2\sqrt{2}&0&0&0&2\sqrt{2}
 \\{}0&-2&0&2&0
 \\{}0&0&-2\sqrt{2}&0& 4
 \end {array}\right ]\,.
 $$
The latter matrix
fits in the classification pattern as provided by
Lemma~\ref{lemma2} above.

\subsection{The first occurrence of $K=3$ at $N=6$}

Besides the trivial, non-degenerate EP6 limit
with $K=1$ we have to consider its anomalous descendants, viz., the
unique $K=2$ decomposition
  $
 \fbox{-5,-3,-1,1,3,5}=\fbox{-3,-1,1,3} \oplus \fbox{-5,5}\,
 $
and the
unique $K=3$ decomposition
  $
 \fbox{-5,-3,-1,1,3,5}= \fbox{-1,1}
  \oplus \fbox{-3,3} \oplus \fbox{-5,5}
 $.
In both of these cases
(listed in illustrative Table \ref{dowe} above)
the
direct-sum components of
$H^{(6)}_{(K=2,K=3)}(\lambda^{(EP6)})$
may be found displayed in Eq.~(\ref{epthisset}).
In the latter case, for example, the direct sum yields the matrix
 $$
  H^{(6)}_{(K=3)}(\lambda^{(EP6)})
 =\left [\begin {array}{rrrccc}
    -5&0&0&0&0&5
 \\{}0& -3&0&0&3&0
 \\{}0&0&-1&1&0&0
 \\{}0&0&-1&1&0&0
 \\{}0&-3&0&0&3&0
 \\{}-5&0&0&0& 0&5
 \end {array}\right ]\,.
 $$
The number of scenarios is $a(6)=3$.
The role
and consequences of small perturbations of the latter matrix
were
analyzed in \cite{anomalous}.

\subsection{Paradox of decrease of $a(N)$ between $N=7$ and $N=8$}

At $N=7$ the number of the eligible EP7 scenarios is $a(7)=6$
because
the usual trivial $K=1$ option can be accompanied by
the following quintuplet of anomalous EP7 direct sums,
  $$
 \fbox{-6,-4,-2,0,2,4,6}=
  \fbox{-4,-2,0,2,4} \oplus \fbox{-6,6}\,,
 \ \ \ \ \ K=2\,,
  $$
  $$
 \fbox{-6,-4,-2,0,2,4,6}=\fbox{-2,0,2} \oplus \fbox{-4,4}
 \oplus \fbox{-6,6}\,,
 \ \ \ \ \ K=3\,,
 $$
  $$
 \fbox{-6,-4,-2,0,2,4,6}=\fbox{-4,0,4} \oplus \fbox{-2,2}
 \oplus \fbox{-6,6}\,,
 \ \ \ \ \ K=3\,,
 $$
  $$
 \fbox{-6,-4,-2,0,2,4,6}=
 \fbox{-4,0,4} \oplus \fbox{-6,-2,2,6}\,,
 \ \ \ \ \ K=2\,,
 $$
  $$
 \fbox{-6,-4,-2,0,2,4,6}=\fbox{-6,0,6} \oplus \fbox{-2,2}
 \oplus \fbox{-4,4}\,,
 \ \ \ \ \ K=3\,.
 $$
In contrast, at
$N=8$
we have $a(8)=4$, i.e.,
only the triplet of the anomalous, $K>1$ direct sums becomes available,
viz.,
  $$
 \fbox{-7,-5,-3,-1,1,3,5,7}=\fbox{-5,-3,-1,1,3,5}
  \oplus \fbox{-7,7}\,,
 \ \ \ \ \ K=2\,,
 $$
  $$
 \fbox{-7,-5,-3,-1,1,3,5,7}=\fbox{-3,-1,1,3} \oplus \fbox{-5,5}
 \oplus \fbox{-7,7}\,,
 \ \ \ \ \ K=3\,,
 $$
  $$
 \fbox{-7,-5,-3,-1,1,3,5,7}= \fbox{-1,1} \oplus \fbox{-3,3}
  \oplus \fbox{-5,5}\,,
  \oplus \fbox{-7,7}\,,
 \ \ \ \ \ K=4\,.
    $$
The latter item is our first four-term direct-sum decomposition
example
representing a
fifteen-diagonal but very sparse matrix
$ H^{(8)}_{(K=4)}(\lambda^{(EP8)})$
with bi-diagonal structure.

\section{Discussion \label{speciae}}


The field of study of
quantum phase transitions
and of the
role played by the Kato's exceptional points
may be characterized by a lasting conflict
between ambition and reality. The enthusiasm accompanying the
production of ideas on the side of theory ({\it pars pro toto\,} let us
mention several older
compact reviews of physics of non-Hermitian degeneracies
\cite{Berry,Heiss})
seems counterbalanced by the difficulties of the search
for EPN-related phase transitions in the laboratory \cite{Stefan}.

The concept of the exceptional-point
value $\lambda^{(EPN)}$ of a real parameter
$\lambda$ in a linear operator
$H(\lambda)$
proved, originally, useful
in mathematics \cite{book,Kato}.
Physics behind the EPNs
remained obscure.
The situation has changed after several authors
discovered
that the concept admits
applicability in multiple branches of quantum as well as non-quantum
physics \cite{Carl,Carlbook,Christodoulides}.
{\it Pars pro toto\,} let us mention that
in the subdomain of quantum physics the values of $\lambda^{(EPN)}$
acquired the status
of instants of an experimentally realizable
quantum phase transition \cite{Geyer,BB,denis}.
The
boundary-of-stability
role played by the values of $\lambda^{(EPN)}$
attracted, therefore, attention of experimentalists \cite{Heiss,Joglekar}
as well as of theoreticians \cite{Trefethen,Dorey,Fabio}.

During our study of the problem we felt challenged by both of these
aspects of the phenomenon.
On the side of experiment we were
impressed, first of all, by an intimate connection
between the mathematics of EPNs and the very concrete physics
of quantum phase transitions \cite{catastb}. On the side of theory
we felt motivated by the existence of its
two mutually interrelated aspects.
The first one
was pragmatic: the
description
of the
processes of the loss of the observability
seems to be hardly feasible by the conventional numerical means
\cite{bhgt6}. At the same time, the perturbation-approximation
techniques appeared to be applicable after minor amendments
\cite{without}.
The second attractive aspect of the theory was conceptual and deeper:
the non-Hermitian degeneracy of
an $N-$plet of the stable bound states with $N>2$
only became tractable as an admissible process
in the framework of
quasi-Hermitian formulation of quantum mechanics \cite{Geyer}
(at present, people more often use the term
${\cal PT}-$symmetric theory,
cf. reviews
\cite{SIGMA,ali,Carlbook,Christodoulides}).

In the latter framework
one encounters
challenges and apparent contradictions
reflecting the
unexpected
peaceful coexistence of the
{non-Hermitian} EPN-related degeneracy
of spectra (which are,
by assumption \cite{Carl,BB}, {real})
with the {unitarity} of evolution.
Another puzzle may be seen
in a rather vague correspondence between
some of the EPN properties and
the structure of the operators, etc.
In our present study we felt guided
by the
contrast between the not too surprising
numerical ill-conditioning \cite{Wilkinson}
and the
real-matrix
nature of the exactly solvable TAO models
of Refs.~\cite{maximal,tridiagonal}.
The over-restrictive
form of these models
was reclassified
as inessential.
We accepted the conjecture
of
correlation between the tridiagonality of
Hamiltonians
and a {triviality} of the geometric
multiplicities \cite{anomalous}. Keeping the warning in mind
we
constructed the universal GAO $K>1$ models
via the TAO-direct-sum ansatzs.


\subsection{Unitary vs. non-unitary systems}

The theoretical background of our present complete
menu of non-numerically tractable
phase-transition scenarios
lies in the consistent
theoretical compatibility of the non-Hermiticity of
the Hamiltonian
with the unitarity of the quantum evolution
in the
quasi-Hermitian Schr\"{o}dinger
picture~\cite{Geyer,Carl,ali,MZbook}.
In order to avoid misunderstandings
we must immediately add that in many
experiment-oriented
descriptions, the
quantum-system transmutations
mediated by the EPNs
{\it alias\,} non-Hermitian degeneracies \cite{Berry}
are very often non-unitary.
In the extensive related literature  \cite{Nimrod}
the scope of the theory is very broad.
In the
Feshbach's
open-system spirit \cite{Feshbach},
people work with the non-Hermitian
effective Hamiltonians $H_{\rm eff}$
with spectra which are complex.
Still, the quantum systems in question are realistic and
their
analysis profits from sharing the
mathematical know-how with the
${\cal PT}-$symmetric theories
of the closed, unitary system.

From the historical perspective,
the open-system
philosophy
was always dominant. For a broader audience
this dominance
has only been shattered, in 1998, by Bender with Boettcher \cite{BB}.
These authors proposed that
at least some of the processes of the
quantum phase transitions
might find a more natural
description and explanation in the alternative,
closed-system theoretical
framework (see reviews
\cite{Carl,ali,book}).

The Bender- and
Boettcher-inspired change of the paradigm had two roots,
both of them related to the problem of quantum phase transitions.
One is that even in many closed
quantum systems the unitary evolution
can be controlled by the Hamiltonian
(with real spectrum)
which is non-Hermitian
(i.e., admitting the EPNs).
Although such a conjecture might sound like a paradox,
the Bender's and Boettcher's
secret trick was that the latter Hamiltonian
can be,
via an
appropriate amendment of the Hilbert space
of states,  Hermitized,
fitting all of the standard postulates
of textbooks
(see, e.g., the older review paper \cite{Geyer}
for some basic mathematical details).

The second root of the change of the paradigm
concerns the
quantum phase transitions more immediately,
opening their innovative treatment
via specific examples.
In \cite{BB}
the innovation was
sampled
by the spontaneous breakdown of
parity-time (i.e., ${\cal PT}$) symmetry.
The authors related
the
collapse
of the system
to the coincidence of the parameter in
$H(\lambda)$ with its
EPN
value
$\lambda^{(EPN)}$.
In the
phase-transition limit
$\lambda \to \lambda^{(EPN)}$ they, indeed, encountered
a genuine qualitative novelty.
In their non-Hermitian
local-interaction
models the degeneracy was mostly followed by an abrupt
complexification, i.e., by
a sudden loss of the
observability~\cite{BB}. Their EPN-related
``quantum catastrophic'' behavior was realized as a
merger
of
bound-state energies (\ref{siesta}).

The key challenge was to show,
for every preselected non-Hermitian Hamiltonian, that
in the real-spectrum regime with
$\lambda < \lambda^{(EPN)}$
the system remains unitary.

\subsection{Correct inner products \label{statartbe}}

For a preselected non-Hermitian quantum
Hamiltonian with real spectrum
a specification of its correct probabilistic
closed-system interpretation
can be a straightforward linear-algebraic procedure,
especially for the most elementary TAO
tridiagonal Hamiltonian matrices (\ref{pentoy})
(cf.
the detailed and constructive
discussion of this point in
paper \cite{recurrent}).
A firm theoretical ground
of such a procedure (explained, e.g., in
review \cite{ali})
is to be sought in an
amendment  of
the conventional Hilbert space.
For the sake of definiteness, we may
denote
the amended,
correct Hilbert space
by dedicated symbol ${\cal H}$,
with the other, manifestly unphysical
but mathematically friendlier
Hilbert space denoted by
a different symbol, say,  ${\cal K}$.

In this notation
the spectrum of any candidate
$H(\lambda)$ for an
observable Hamiltonian must be kept real. Although such an
operator must be selfadjoint  in ${\cal H}$
(i.e.,
in the physical Hilbert space,
obedient to the Stone theorem \cite{Stone}),
it will be, in general,
non-Hermitian in ${\cal K}$.
The key purpose of such a duplicity
(often attributed to Dyson \cite{Dyson})
is that
via transition to a user-friendlier space ${\cal K}$,
one achieves a decisive
simplification of all calculations.
In parallel, in the language of physics one
transfers the responsibility for the
unitarity of the system from
the $\lambda-$dependence of
the Hamiltonian
to the more flexible $\lambda-$dependence of
the Hilbert space ${\cal H}={\cal H}(\lambda)$.
Such a ``double picture'' of evolution
offers a highly welcome physical
model-building freedom while still guaranteeing that
the evolution generated by  ${H}(\lambda)$
remains unitary in ${\cal H}(\lambda)$.

In the majority of applications
(including the one used in the present paper)
the original Dyson's flowchart of the theory
is inverted. The model-building process
is initiated by the choice of a non-Hermitian ${H}(\lambda)$
acting in a $\lambda-$independent
Hilbert space ${\cal K}$ endowed
with a conventional
inner product $\langle \psi_1|\psi_2\kt_{\cal K}$
which is unphysical but user-friendly.
Naturally, what is then needed is a reconstruction of
the ``missing'' Hilbert space ${\cal H}(\lambda)$ \cite{ali}.

In a way described in \cite{Geyer} the construction
of ${\cal H}(\lambda)$ becomes significantly facilitated
when the candidate for the Hamiltonian is finite-dimensional,
${H}(\lambda)={H}^{(N)}(\lambda)$.
Then, there will exist multiple Hermitian and positive-definite
matrices $\Theta=\Theta^{(N)}(\lambda)$
which satisfy the $N$ by $N$ matrix equation
 $$
 [{H}^{(N)}(\lambda)]^\dagger\,\Theta^{(N)}(\lambda)
 =\Theta^{(N)}(\lambda)\,{H}^{(N)}(\lambda)\,.
 $$
In terms of any one of these matrices
one can, subsequently, define the space ${\cal H}(\lambda)$
via the mere redefinition of the inner product in  ${\cal K}$,
 \be
 \langle \psi_1|\psi_2\kt_{\cal H}
 =\langle \psi_1|  \Theta|\psi_2\kt_{\cal K}\,.
 \label{innp}
  \ee
The
differences between the two alternative representation
spaces is in fact reduced to the mere non-equivalence of
the respective inner products.
Due to such an elementary
mathematical correspondence the standard textbook description of
the unitary evolution dynamics defined in the difficult
Hilbert space ${\cal H}$
finds a simplification
in which all of the
necessary calculations and predictions
are assumed to be made
in the much simpler representation space ${\cal K}$
(i.e., in our present paper, in the
most elementary Euclidean real vector space
${\cal K}=\mathbb{R}^{N}$).

\subsection{The corridors of unitarity\label{tridva}}

Whenever one keeps the evolution
unitary, the values of $\lambda^{(EPN)}$
mark the
points of the loss of the observability
of the system \cite{catast,sdenisem}.
Our present
hierarchy of specific GAO models may be
treated as certain
exactly solvable quantum analogue of the Thom's  typology of classical
catastrophes \cite{Zeeman},
with potential applicability to closed as well as open systems.

In our present paper we were exclusively
interested in the former type of applications.
We knew that
the dynamics
of any unitary quantum system
(i.e., typically, its stability with respect to small perturbations)
is strongly influenced by the
EPNs. We should only add that an extreme care must be paid
to
the Stone theorem \cite{Stone}
requiring the
Hermiticity of
$H(\lambda)$
in the related physical Hilbert space ${\cal H}$.
A Hermitization of the Hamiltonian is needed
\cite{Geyer}.
Such a process involves a
reconstruction of an
appropriate amended inner product
in the conventional but
unphysical Hilbert space ${\cal K}$.
Interested readers may find
one of the rare samples of
such a reconstruction of the whole menu of ${\cal H}$s
in \cite{scirep}.

In the realistic models one often
encounters a paradox that the construction of
the correct, amended inner product
may happen to be prohibitively complicated, i.e.,
from the pragmatic point of view,
inaccessible.
This is the reason why people often
postpone the problem and use,
temporarily, a simplified inner product.
For our present, user-friendly, matrix-represented
GAO Hamiltonians of closed systems with $N < \infty$
such a purely technical obstacle does not occur.
The
reconstruction  of ${\cal H}$
would be a routine
application of linear algebra.
Reclassifying the real GAO matrices which are
non-Hermitian in our auxiliary space
${\cal K}=\mathbb{R}^N$ (that's why
we write $H \neq H^\dagger$)
into operators which are, by construction, Hermitian in ${\cal H}$
(in \cite{MZbook}, for example, we wrote $H=H^\ddagger$).

An unusual
property
of the physical, amended Hilbert space is that it is
Hamiltonian- and
$\lambda-$dependent, ${\cal H}={\cal H}(\lambda)$.
In the present phase-transition context the
Hamiltonian $H(\lambda)$ itself may be interpreted as
Hermitian just for
``admissible'', unitarity-compatible
$\lambda$s forming a domain ${\cal D}$.
The system is able to reach, via unitary evolution,
the instant of the EPN-related quantum phase transition
if and only if the overlaps of ${\cal D}$
with the arbitrarily small vicinities
of $\lambda^{(EPN)}$ remain all non-empty
forming a ``corridor of access''
${\cal D}^{(EPN)}$  \cite{corridors}.

In paper \cite{corridors} it has been shown
that
the corridors ${\cal D}^{(EPN)}$
connecting,
in the space of matrix elements,
the EPN extremes
with the points in a deep interior of  ${\cal D}$
do always exist. For the
TAO matrices (\ref{epthisset}), in particular,
the shape of these corridors
has been found, in \cite{tridiagonal}, sharply spiked.
For the present broader class of the $K>1$ models
the existence of the analogous corridors
of the unitary access to EPNs
has been conjectured in \cite{degdeg}.

\section{Summary\label{summary}}

The phenomenological usefulness
of the
TAO models of
Refs.~\cite{passage,maximal,tridiagonal}
re-appears, unexpectedly,
also in the $K>1$ GAO-based phase-transition context.
We showed that in the amended theory
the TAO matrices
may be assigned the role of
building blocks, and that such a trick implies,
as one of its its byproducts, the
exact solvability of the
resulting benchmark GAO quantum systems
at any $K>1$ and $\lambda \in (0,\lambda^{(EPN)})$.
Along these lines,
a universal mathematical classification as well as
a richer structure of
predictions of the measurable phenomena is achieved.

A menu of eligible benchmark EPN-supporting specific models
is proposed and described as
controlled by alternative multidiagonal $N$ by $N$ matrix
realizations
of the direct-sum
anharmonic-oscillator-type Hamiltonians.
On methodical level the
phenomenological non-equivalence
of these partially
$K-$related decompositions
of the Hamiltonians
is to be stressed.

Via our systematic explicit
enumeration of the most elementary special cases
we emphasized, first of all,
the importance of the
proper treatment of the geometric multiplicity $K$
of the $N-$plets of energy
levels in the EPN-related phase-transition dynamical regime.
We showed that
in a pre-critical stage of the transition
the nontrivial multiplicities $K>1$
just reflect the existence of the phenomenon of
a $K-$tuple clusterization of the wave functions
of bound states.
This may be of interest in experiments in which
the processes of the loss of observability
are currently being studied, mostly
due to the existing technical limitations, just at the
small $N$ and trivial $K=1$.
Our present results may offer a motivation
for performing some extended and subtler analyses
or simulations of the clusterized $K>1$  processes in the laboratory.

\subsection*{Acknowledgments}

The author acknowledges the financial support from the
Excellence project P\v{r}F UHK 2021.


\section*{Appendix A: Direct-sum-decomposition partitionings ${\cal R}(N)$}

In the present GAO-based phenomenological models
numbered by the Hamiltonian-matrix dimensions $N=(1), 2,3,\ldots$,
the counts of the eligible non-equivalent
EPN-related dynamical scenarios (i.e.,
direct-sum decompositions (\ref{limitsN}) and  (\ref{limitsNK})
as sampled, in Table \ref{dowe}, at $N=6$)
form a sequence
 \be
 a(N) = (0), 1, 1, 2, 3, 3, 6, 4, 11, 6, 17, 7, 32, 8, 47, 13, 66,
  \ldots\,.
 \ee
The evaluation of the sequence
is important and useful for at least two reasons. First,
beyond the smallest $N$,
it enables us to check the completeness
of the EPN-related
dynamical alternatives.
Second,
the asymptotically exponential growth of the sequence
indicates that at the larger $N$s, the
menus of the EPN-supporting toy models
numbered by partitionings ${\cal R}(N)$
will be dominated
by the anomalous, multidiagonal $K>1$ Hamiltonians.

Besides that, the properties of the sequence
are of an independent mathematical interest.
First of all we notice that our sequence
seems composed of the two apparently simpler,
monotonously increasing integer subsequences.
They have to be discussed separately.

\subsection*{A.1. Subsequence of $a(N)$ with even $N=2J$, $J=1,2,\ldots$.}

The values
 \be
 b(J)=a(2J) =  1, 2, 3, 4, 6, 7, 8, 13, 14, 15, 25, 26, 33, 50, \ldots
 \ee
of the even-dimension (sub)sequence
may be generated by the algorithm described in \cite{oeiseven} and
carrying the identification code number
A336739.
Recalling this source let us summarize a few key
mathematical features of the sequence.

\begin{defn}
\label{ace}
The quantity b(n) is the number of
decompositions of {B}(n,1) into disjoint unions of {B}(j,k)
where {B}(j,k) is the set of numbers \{(2\,i-1)\,(2\,k-1), 1
$\leq$ i $\leq$ j \}.
\end{defn}


 \noindent
It may be instructive to display a few examples:

{B}(n,1) are the sets \{1\}, \{1,3\}, \{1,3,5\}, \{1,3,5,7\}, ...,

{B}(n,2) are the sets \{3\}, \{3,9\}, \{3,9,15\}, \{3,9,15,21\}, ...,

{B}(n,3) are the sets \{5\}, \{5,15\}, \{5,15,25\}, \{5,15,25,35\}, ...,

 \noindent
etc. There are two decompositions of {B}(2,1) = \{1,3\}, viz.,
trivial {B}(2,1) and nontrivial {B}(1,1) + {B}(1,2) = \{1\} + \{3\}.
Similarly, the complete list of the $a(5) = 6$
decompositions of \{1,3,5,7,9\}
is as follows:

  \{\{1,3,5,7,9\}\},

  \{\{1,3,5,7\}, \{9\}\},

  \{\{1,3,5\}, \{7\}, \{9\}\},

  \{\{1,3\}, \{5\}, \{7\}, \{9\}\},

  \{\{1\}, \{3\}, \{5\}, \{7\}, \{9\}\},

  \{\{3,9\}, \{1\}, \{5\}, \{7\}\}.

 \noindent
We should add that
the
notation used in definition \ref{ace} of quantities  {B}(j,k) is
mathematically optimal.
For the purposes of our present paper, nevertheless, it is
necessary to recall the equivalence of every {B}(j,k)
to one of the present boxed symbols. For example,
in place of
 {B}(3,1) = \{1,3,5\} we should write ${\cal B}$[3,1] =
  \fbox{-5,-3,-1,1,3,5}, etc.
Definition \ref{ace} can be modified as follows.

\begin{defn}
\label{bece}
The quantity b(n) is the number of different decompositions
of ${\cal B}$[n,1] into
unions of ${\cal B}$[j,k] where ${\cal B}$[J,K] is defined as
the boxed symbol \\
. \hspace{2cm}
 \fbox{(2K-1)(1-2J),(2K-1)(3-2J),(2K-1)(5-2J), ... ,(2K-1)(2J-1)}.
\end{defn}

\subsection*{A.2.  Subsequence of $a(N)$ with odd $N=2J+1$, $J=1,2,\ldots$.}


The values of the subsequence
 \be
 c(J)=a(2J+1) =   1, 3, 6, 11, 17, 32, 47, 66, 105, 162, 198, 376,   \ldots
 \ee
may be found discussed
 in \cite{oeisodd}.
 Using this source
let us summarize a few key aspects of this sequence
which carries the identification number A335631.

\begin{defn}
The quantity c(n) is the number of decompositions of {C}(n,1) into
disjoint unions of {C}(j,k) and G(q,r) where {C}(j,k) is the set of
numbers \{i\,k, 0 $\leq$ i  $\leq$ j \} and where G(q,r) is the
set of numbers \{(2\,p-1)\,r, 1  $\leq$ p  $\leq$ q \}.
\end{defn}

 \noindent
In a more explicit manner let us point out that

{C}(n,1) are the sets \{0,1\}, \{0,1,2\}, \{0,1,2,3\}, \{0,1,2,3,4\}, ...,

{C}(n,2) are the sets \{0,2\}, \{0,2,4\}, \{0,2,4,6\}, \{0,2,4,6,8\}, ...,

{C}(n,3) are the sets \{0,3\}, \{0,3,6\}, \{0,3,6,9\}, \{0,3,6,9,12\}, ...,

 \noindent
etc., and that

G(n,1) are the sets \{1\}, \{1,3\}, \{1,3,5\}, \{1,3,5,7\}, ...,

G(n,2) are the sets \{2\}, \{2,6\}, \{2,6,10\}, \{2,6,10,14\}, ...,

G(n,3) are the sets \{3\}, \{3,9\}, \{3,9,15\}, \{3,9,15,21\}, ...,

 \noindent
etc. We can say that
$a(2) = 3$ because the decompositions of {C}(2,1) = \{0,1,2\}
involve not only the trivial copy {C}(2,1) but also the nontrivial
formulae
{C}(1,2) + G(1,1) = \{0,2\} + \{1\} and
{C}(1,1) + G(1,2) = \{0,1\} + \{2\}.
Similarly: why $a(3) = 6$?
Because the decompositions of \{0,1,2,3\} are as follows:

  \{\{0,1,2,3\}\},

  \{\{0,1,2\}, \{3\}\},

  \{\{0,1\}, \{2\}, \{3\}\},

  \{\{0,2\}, \{1,3\}\},

  \{\{0,2\}, \{1\}, \{3\}\},

  \{\{0,3\}, \{1\}, \{2\}\}.

 \noindent
The one-to-one correspondence and the translation of this notation
to our present boxed-symbol language is again obvious, fully
analogous to
the one described in the preceding paragraph.


\end{document}